\documentclass[11]{article}
\usepackage[english]{babel}
\usepackage[a4paper, total={5.7in, 10in}]{geometry}
\usepackage{algorithm}
\usepackage{algpseudocode}
\usepackage{amsmath}
\usepackage{amssymb}
\usepackage{mathtools}
\usepackage{authblk}
\usepackage{graphicx}
\usepackage{caption}
\usepackage{float}
\usepackage{microtype}
\usepackage{setspace}
\usepackage{booktabs}
\newtheorem{theorem}{Theorem}
\newtheorem{definition}{Definition}
\newenvironment{proof}
{\textit{Proof:} }
{$\square$}

\DeclareMathOperator*{\argmax}{arg\,max}

\usepackage[colorlinks=true, allcolors=blue]{hyperref}

\title{Towards a power analysis for PLS-based methods}
\author[1,*]{Angela Andreella}
\affil[1]{\small Department of Economics, Ca' Foscari University of Venice, Venice, Italy}
\author[3]{Livio Finos}
\author[3]{Bruno Scarpa}
\author[2]{Matteo Stocchero}
\date{}
\affil[2]{\small Department of Women's and Children's Health, University of Padova, Padova, Italy}
\affil[3]{\small Department of Statistical Sciences, University of Padova, Padova, Italy}
\affil[*]{\small Corresponding author: angela.andreella@unive.it}
\providecommand{\keywords}[1]{\textbf{Keywords:} #1}

\begin{document}
\maketitle
\sloppy 

\begin{abstract}
In recent years, power analysis has become widely used in applied sciences, with the increasing importance of the replicability issue. When distribution-free methods, such as Partial Least Squares (PLS)-based approaches, are considered, formulating power analysis turns out to be challenging. In this study, we introduce the methodological framework of a new procedure for performing power analysis when PLS-based methods are used. Data are simulated by the Monte Carlo method, assuming the null hypothesis of no effect is false and exploiting the latent structure estimated by PLS in the pilot data. In this way, the complex correlation data structure is explicitly considered in power analysis and sample size estimation. The paper offers insights into selecting statistical tests for the power analysis procedure, comparing accuracy-based tests and those based on continuous parameters estimated by PLS. Simulated and real datasets are investigated to show how the method works in practice.
\end{abstract}
\keywords{Classification, Partial Least Squares, Permutation Tests, Power Analysis}

\section{Introduction}

Nowadays, scientists are increasingly forced by ethical and economic considerations to apply power analysis for sample size estimation since inferential statistics can only lead to robust and reliable results by implementing the right experimental design. Indeed, the number of observations should not be too large for efficiency, ethical, and cost reasons but enough to guarantee reliable statistical results with minimal false positive or false negative rates. Several authors highlighted that the results of many published biomedical studies are unreliable and probably false due to the small sample size used \cite{Button2013, Ioannidis2005}, and most of the published -omics studies are underpowered, as declared by the authors themselves. 

In particular, when analyzing multivariate data, the responses are typically correlated, redundant, and noisy, and the number of observations is generally smaller than the number of responses. In this framework, likelihood-free approaches such as Partial Least Squares (PLS)-based methods \cite{Wold1983,wold2001pls} are applied, and traditional techniques for power analysis cannot be applied to sample size estimation. In addition, methods such as PLS (as well as canonical correlation analysis (CCA) \cite{jordan1875essai, hotelling1992relations} and principal component analysis (PCA) \cite{pearson1901liii, jolliffe2002principal}) are generally viewed as exploratory methods rather than as testing procedures \cite{winkler2020permutation}. This means that the model parameters are uncommonly interpreted as statistical effects, and no single definition for the effect size is available. For these reasons, we propose here a methodology for performing power analysis when PLS-based methods are used as tools for data analysis. 

Since PLS-based methods are not based on statistical distributions and are not likelihood estimation-based techniques, strategies for power analysis and sample size estimation should be then based on numerical simulation. One method common in the literature is using Monte Carlo (MC) simulation \cite{Martens2000,muthen2002use} that generates a large set of artificial data using the design under evaluation to design hypothetical experiments. The artificial data are analyzed in the same way as the real dataset, obtaining the distributions of the model parameters of interest. Then, the optimal sample size is estimated by studying the cost and risks of Type I and Type II errors associated with the given experimental design.


To the best of our knowledge, only a few studies investigated the problem of sample size estimation in the case of PLS-based techniques. Blaise and coauthors \cite{Blaise2016} introduced an interesting approach based on MC simulation where the correlation between variables is explicitly incorporated. Specifically, new samples with marginal distributions and correlation structure similar to the ones observed in the pilot data are simulated, modeling the log-transformed pilot data as a multivariate normal distribution. Thus, the effect size is introduced by acting on subsets of correlated variables. The relationship between statistical power, sample size, and effect size is investigated by analyzing the artificial data and obtaining the distributions of the statistics of interest. In a second study, Saccenti and Timmerman \cite{Saccenti2016} drew a possible line of thought to perform power analysis for PCA and PLS for Discriminant Analysis (PLS-DA). Important concepts have been discussed, and some interesting ideas have been offered for sample size estimation in a multivariate setting. Specifically, in the case of PCA, they proved that inference and sample size estimation could be grounded by solid statistical characterization of the distributional properties of the PCA solution, while for PLS-DA, the scenario is more complex. In another study, Nyamundanda and coauthors \cite{Nyamundanda2013} proposed a general method for sample size estimation based on simulated data generated from Probabilistic Principal Component Analysis (PPCA)-based models that can be applied without experimental pilot data. The approach considers only univariate data analysis controlling false discovery rate as a data analysis strategy and does not examine PLS-based methods.

In this manuscript, we present a new method for power analysis that uses the score structure discovered by PLS from the pilot data to simulate new datasets with the same covariance structure of the pilot data but different sample sizes. This is one of the main novelties of our study. Indeed, unlike the above-mentioned published approaches, data simulation used in power analysis is performed here to see the data from a multivariate point of view, respecting the correlation structure. The same strategy PLS uses to model the data is then applied in the data simulation process. This makes the approach tailored explicitly for PLS. Another novelty with respect to the above-mentioned published methods is that different statistical tests for the power analysis procedure are proposed and investigated. Moreover, the approach is developed as \texttt{R} package to ensure accessibility for researchers, enhancing transparency and replicability of the results. 

For the sake of simplicity, we consider the simple case of a case-control study that is investigated employing PLS for classification. The reason for focusing our attention on this simple case is twofold. First, the case-control setting is one of the most largely implemented study designs in -omics sciences and beyond. Second, PLS for classification in the limit case of a $2$-class problem is a simple but not trivial example of a PLS-based technique that can be used both to show how our new approach can be formulated and to discuss the use of different statistical tests. Regarding the statistical tests, three statistical tests are considered to analyze the null hypothesis of equal distribution between the two classes. The first statistical test is based on the Matthew Correlation Coefficient (MCC), widely used in the PLS analysis framework. Since this statistical test loses power due to its discretization nature \cite{rosenblatt2021better}, we propose a two-sample t-test based on the predictive score matrix estimated by PLS and the squared Pearson correlation coefficient properly transforming the categorical dependent variable.

The proposed method can be extended to other PLS-based techniques following the same lines of thought. 

The paper is structured as follows. Section \ref{background} summarizes the theory underlying PLS to provide a background to understand the following sections. Section \ref{mcc} defines the permutation-based statistical tests to analyze the null hypothesis of equal distribution between the two classes. Section \ref{power} shows the strategy for power analysis and sample size estimation based on MC simulation. Section \ref{applications} analyzes simulated and real datasets to show how the method works in practice. Discussion and concluding remarks are reported in Section \ref{discussion}.

Without loss of generality, we consider that the data are mean-centered throughout the manuscript unless stated otherwise.

\section{Background}\label{background}
The theoretical framework of PLS for regression (PLSR) \cite{Wold1983} and PLS for classification (PLSc) \cite{stocchero2021pls} is briefly presented to give the reader the helpful background to understand the procedure for power analysis and sample size estimation introduced in Sections \ref{mcc} and \ref{power}.

\subsection{PLS for regression}\label{pls}

Let $\mathbf{Y} \in \mathbb{R}^{N \times K}$ be the matrix of dependent variables and $\mathbf{X} \in \mathbb{R}^{N \times P}$ the matrix of covariates, where $N$ is the number of observations, $K$ the number of dependent variables, and $P$ the number of independent ones. Considering the linear regression model
\begin{equation*}
    \mathbf{Y} = \mathbf{X} \mathbf{B} + \mathbf{F}
\end{equation*}
where $\mathbf{F} \sim (0, \mathbf{\Sigma})$ is the error term matrix and $\mathbf{\Sigma}$ the corresponding covariance matrix, PLSR estimates the coefficient matrix $\mathbf{B}$ decomposing the $\mathbf{X}$ and $\mathbf{Y}$ matrices by means of the scores matrix $\mathbf{T}$, as defined in the following definition.

\begin{definition}\label{bifact}
Let be $\mathbf{Y} \in \mathbb{R}^{N \times K}$ and $\mathbf{X} \in \mathbb{R}^{N \times P}$. PLSR can be rephrased in terms of the score matrix $\mathbf{T}$ as:
\begin{align*}
    \mathbf{X} &= \mathbf{T} \mathbf{P}^\top + \mathbf{E} \\
    \mathbf{Y} &= \mathbf{T} \mathbf{Q}^\top + \mathbf{F}
\end{align*}
where $\mathbf{E}$, $\mathbf{F}$ are error terms matrices, and $\mathbf{P} = \mathbf{X}^\top \mathbf{T}(\mathbf{T}^\top \mathbf{T})^{-1} $ and $\mathbf{Q} = \mathbf{Y}^\top \mathbf{T}(\mathbf{T}^\top \mathbf{T})^{-1} $ the loadings matrices of $\mathbf{X}$ an $\mathbf{Y}$, respectively. 
\end{definition}

The scores matrix is calculated by an iterative procedure, the so-called PLS2 algorithm. At each iteration $a \in \{1, \dots, A\}$ with $A \in \{1, \dots, \text{rank}(\mathbf{X})\}$ of the algorithm, a suitable weight vector $\mathbf{w}_a$ is calculated as solution of the following eigenvalue problem
\begin{equation*}
    \hat{\mathbf{E}}_{a-1}^\top \hat{\mathbf{F}}_{a-1} \hat{\mathbf{F}}^\top_{a-1} \hat{\mathbf{E}}_{a-1} \mathbf{w}_a = \lambda_a \mathbf{w}_a 
\end{equation*}
where $\hat{\mathbf{E}}_{a-1}$ and $\hat{\mathbf{F}}_{a-1}$ are the residual matrices calculated in the previous iteration $a-1$ and $\lambda_a$ is the eigenvalue associated with $\mathbf{w}_a$. The weight vector is then used to project the residual matrix of $\mathbf{X}$ to obtain
\begin{equation*}
\mathbf{T} = [\mathbf{t}_a] = [ \hat{\mathbf{E}}_{a-1} \mathbf{w}_a ] \in \mathbb{R}^{N \times A}.
\end{equation*}
We have used $[\cdot]$ to stand the generic column vector of the referred matrix.

At the first iteration, $\hat{\mathbf{E}}_0 \coloneqq \mathbf{X}$ and $\hat{\mathbf{F}}_0 \coloneqq \mathbf{Y}$ whereas, after $A$ iterations, the final residual matrices are $\hat{\mathbf{E}}\coloneqq \hat{\mathbf{E}}_A$ and $\hat{\mathbf{F}}\coloneqq \hat{\mathbf{F}}_A$. The complete PLS2 algorithm is reported in Appendix \ref{appendix:AppendixA}.

It is worth noting that $\mathbf{T}$ is a linear combination of the columns of $\mathbf{X}$ and that the complexity of the model depends on the number $A$ of iterations, which is then extremely important in PLS. Indeed, given the number $A$ of iterations, the PLS model is completely defined. Moreover, the number $A$ is also the rank of the score matrix $\mathbf{T}$ and of the weight matrix $\mathbf{W}$, having both orthogonal columns \cite{Hoskuldsson1988,stocchero2019iterative}.

After $A$ iterations, the matrix of the regression coefficients $\mathbf{B}$ is then estimated as $\hat{\mathbf{B}} = \mathbf{W}(\mathbf{W}^\top \mathbf{X}^\top \mathbf{X} \mathbf{W})^{-1} \mathbf{W}^\top \mathbf{X}^\top \mathbf{Y}$ and $\mathbf{W} = [\mathbf{w}_a] \in \mathbb{R}^{P \times A}$. In the case of a full column rank matrix of the covariates, it can be shown that $\hat{\mathbf{B}}$ is a biased estimator of $\mathbf{B}$ since $\mathbb{E}(\hat{\mathbf{B}}) = \mathbf{W} (\mathbf{W}^\top \mathbf{X}^\top \mathbf{X} \mathbf{W})^{-1} \mathbf{W}^\top \mathbf{X}^\top \mathbf{X} \hat{\mathbf{B}}_{\text{OLS}}$ where $\hat{\mathbf{B}}_{\text{OLS}}$ is the Ordinary Least Squares (OLS) estimator. In addition, $\hat{\mathbf{B}} = \hat{\mathbf{B}}_{\text{OLS}}$ if the number of score components equals the number of covariates (i.e., $A = P$). In the general case of a rank-deficient matrix of the covariates, when the maximum number of iterations is performed, the matrix of the regression coefficients is $\hat{\mathbf{B}} = \mathbf{V}_{\mathbf{X}} \mathbf{S}^{-1}_{\mathbf{X}} \mathbf{U}^\top_{\mathbf{X}} \mathbf{Y}$, where the Singular Value Decomposition (SVD) $\mathbf{X} = \mathbf{U}_{\mathbf{X}} \mathbf{S}_{\mathbf{X}} \mathbf{V}_{\mathbf{X}}$ has been considered and solves the least squares problem. Pulling $\hat{\mathbf{B}}$ away from the least squares solution helps balance the trade-off bias-variance, leading to approximations that better predict new observations.

The columns of $\mathbf{T}$ can be used as coordinates to represent the observations in a space with dimension $A$. Since $A\ll P$ in most cases, PLSR produces an efficient data reduction that simplifies the investigation of the data variation of $\mathbf{X}$ explaining $\mathbf{Y}$. Unfortunately, $A$ is often greater than $\text{rank}(\mathbf{Y})$, and then, the matrix factorization of $\mathbf{Y}$ of Defition \ref{bifact} becomes sub-optimal, i.e.,  a higher number of dimensions is being used to represent the data than is actually necessary. To overcome this sub-optimality, \cite{stocchero2016post} proposed an alternative matrix factorization by applying a post-transformation procedure.

Post-transformation is a procedure that, starting from a PLS model, generates a new PLS model where the score space is partitioned into two orthogonal sub-spaces. The first one is the predictive subspace spanned by the predictive score matrix called $\mathbf{T}_{P}$ correlated to the dependent variables. The second one is the non-predictive subspace described by the non-predictive score matrix called $\mathbf{T}_{O}$ orthogonal to $\mathbf{Y}$. So, the number of predictive score vectors equals $\text{rank}(\mathbf{Y})$, and the data variation explaining the dependent variables is included exclusively in the predictive part of the model. Post-transforming the PLSR model returns the following matrix factorization:
\begin{align*}
    \mathbf{X} &= \mathbf{T}_{P} \mathbf{P}_{P}^\top+ \mathbf{P}_{O} \mathbf{T}_{O}^\top+ \mathbf{E}\\
    \mathbf{Y} &=  \mathbf{T}_{P} \mathbf{Q}_{P}^\top+ \mathbf{F}
\end{align*}
where $\mathbf{P}_{P}$ and $\mathbf{P}_{O}$ are respectively the predictive and orthogonal loadings matrices of $\mathbf{X}$, and $\mathbf{Q}_{P}$ the predictive loading matrix of $\mathbf{Y}$. Interestingly, $\hat{\mathbf{E}}$ and $\hat{\mathbf{F}}$, and $\mathbf{\hat{B}}$ are the same of the original PLSR model. The set of predictive scores can then be used as bases to build the latent variables explaining the dependent variables. More details about the post-transformation procedure can be found in Appendix \ref{appendix:AppendixB}.

Most of the techniques of the PLS family can be generated from the PLS2 algorithm used to perform PLSR by modifying the equation for the weight calculation (e.g., including constraints the orthogonally Constrained PLS version is obtained \cite{Stocchero2018constraints}), or introducing suitable dependent variables (e.g., coding the categorical variable with dummy variables, PLS-discriminant analysis \cite{barker2003partial} is obtained). Moreover, all those techniques that are based on the Iterative Deflation Algorithm (IDA) (see Appendix \ref{appendix:AppendixA}) can be post-transformed \cite{stocchero2019iterative}.

\subsection{PLS for classification}\label{plsc}

PLS for classification (PLSc) \cite{stocchero2021pls} is an adaptation of PLSR when the support of $\mathbf{Y}$ equals $\mathcal{Y} = \{1, \dots, G\}$ (i.e., a $G$-class problem with $G$ number of classes). For the sake of simplicity, we consider here the case of a $2$-class problem. However, the approach can be easily extended to the case of $G >2$. 

Given $N_1$ observations of class $1$ and $N_2$ observations of class $2$ such that $N = N_1 + N_2$, we define the block-wise probability-data matrix having dimension $N \times 2$ as:
\begin{equation*}
    \mathbf{Z} = \begin{bmatrix}
    (1-\epsilon) \mathbf{1}_{N_1} & \epsilon \mathbf{1}_{N_1} \\
    \epsilon \mathbf{1}_{N_2} & (1-\epsilon) \mathbf{1}_{N_2} \\
    \end{bmatrix}
\end{equation*}
where $\mathbf{1}_{d}$ is a vector of $d$ ones, and $\epsilon < 1/2$. Mean centering $\mathbf{Z}$ by its mean $\bar{\mathbf{Z}}$ and applying the isometric-log ratio transformation $\text{ilr}(\cdot)$, the vector
\begin{equation}\label{numeric_y}
\mathbf{f}_0=\text{ilr}(\mathbf{Z} \ominus \bar{\mathbf{Z}})
\end{equation}
is obtained. The symbol $\ominus$ indicates the subtraction in the simplex $S^2 = \{[\mathbf{Z}^\top_n] \in \mathbb{R}^{2}: \mathbf{Z}_{ng} >0, \sum_{g = 1}^{2} \mathbf{Z}_{ng} = 1 \}$, where $\mathbf{Z}_{ng}$ is the generic elements in row $n$ and column $g$ of the probability-data matrix $\mathbf{Z}$. In other words $\mathbf{f}_0 = (\mathbf{Z} \ominus \bar{\mathbf{Z}}) \mathbf{H}^\top$ where $\mathbf{H}$ is a $G-1 \times G$ orthonormal matrix where rows are orthogonal to $\mathbf{1}_G$ vector of ones \cite{tsagris2011data}. The PLSc can be formulated as the regression problem
\begin{equation*}
    \mathbf{f}_0 = \mathbf{X} \mathbf{B} + \mathbf{F}.
\end{equation*}

The matrix of the estimated regression coefficients $\hat{\mathbf{B}}$ is obtained using the PLS2 algorithm explained in Subsection \ref{pls} considering $\mathbf{f}_0$ defined in Equation \eqref{numeric_y} instead of $\mathbf{Y}$. The $2$ dimensional probability-data vector for a given observation $\mathbf{x}_n \in \mathbb{R}^{P \times 1}$ with $n \in \{1, \dots, N\}$ is calculated by
\begin{equation*}
    [\hat{\mathbf{Z}}^\top_n] = \text{ilr}^{-1} (\mathbf{x}_n^\top \hat{\mathbf{B}}) \oplus [\bar{\mathbf{Z}}^\top_n]
\end{equation*}
where $\oplus$ stands for the addition in the simplex $\mathcal{S}^2$. Finally, the estimated class membership $\hat{g}_n$ for a given observation $n \in \{1, \dots, N\}$ is the arguments of the maxima of $[\hat{\mathbf{Z}}^\top_n]$, i.e., 
\begin{equation*}
   \hat{g}_n =  \argmax_{g} \lim_{\epsilon \to 0^+} [\hat{\mathbf{Z}}^\top_n]
\end{equation*}
which is independent of the value of $\epsilon$ \cite{stocchero2021pls}.

The model can be post-transformed by applying the same procedure presented for the PLS2 algorithm using $\mathbf{f}_0$ instead of $\mathbf{Y}$. Specifically, post-transformation returns a single vector of predictive scores in the case of a 2-class problem, independently of the transformation used to map the class into the Euclidean space. As a general result, post-transforming the PLSc model leads to $G-1$ predictive score components for a given $G$-class problem.

\section{Statistical tests for PLS-based methods}\label{mcc}

For the sake of simplicity, we consider here again a 2-class problem, but the methodology can be extended to more complex problems. Let denote with $\mathbf{X}_g$ the matrix of the covariates regarding the $g \in \{1,2\}$ class, and with $\mathcal{F}$ and $\mathcal{G}$ the distributions of $\mathbf{X}_1$ and $\mathbf{X}_2$, respectively. Within the PLSc framework, the null hypothesis $\mathcal{H}_0: \mathcal{F} = \mathcal{G}$ is tested to evaluate if the covariates are similarly distributed between the two classes. In principle, various statistical tests can be used to test $\mathcal{H}_0: \mathcal{F} = \mathcal{G}$. Here, we propose three different statistical tests that are then analyzed in terms of power in Section \ref{applications}. 

The first is based on the MCC, equivalent to the normalized Pearson $\chi^2$ statistic, calculated considering the contingency table obtained with the real class, a common choice in PLS literature. Due to the discretization nature of this statistical test, the MCC-based test suffers from low power, particularly if the statistical test is calculated by cross-validation in the case of small sample size. Following the suggestions of Rosenblatt and coauthors \cite{rosenblatt2021better}, we consider the MCC-based test calculated considering the full data (i.e., without cross-validation with or without replacement).

The second statistical test proposed here is based on the predictive score vector $\mathbf{T}_P$. Under $\mathcal{H}_0: \mathcal{F} = \mathcal{G}$, we have $\mathcal{T}_{1P} = \mathcal{T}_{2P}$ where $\mathcal{T}_{gP} \in \mathbb{R}^{N_g \times 1}$ is the distribution of the predictive scores considering the class $g \in \{1,2\}$. The statistical test is defined as a two-sample t-test considering as samples the predictive scores for each class $g$ as samples. The statistic is, in principle, more powerful than the one based on the MCC because it overpasses the discretization problem, as it will be seen in the simulation analysis presented in Section \ref{simulated data}.

The third one is the squared Pearson correlation coefficient $R^2$ between the observed dependent variable $\mathbf{f}_0=\text{ilr}(\mathbf{Z} \ominus \bar{\mathbf{Z}})$ defined in Equation \eqref{numeric_y} and the estimated one $\mathbf{X} \hat{\mathbf{B}}$. As per the previous statistics, $R^2$ is, in principle, more powerful than MCC overpassing the discretization problem, even if both $R^2$ and MCC are based on the estimated matrix of the regression coefficients.

Let denote with $\mathcal{S}$ one of the three statistical tests proposed above. We rely on its permutation distribution to compute the corresponding $p$-values. Let define with $\mathcal{P}$ the set of all possible permutation matrices; we randomly select $J$ permutation matrices $\mathbf{P}_j \in \mathcal{P}$ where $1 \le j \le J\le |\mathcal{P}|$. Since under $\mathcal{H}_0: \mathcal{F} = \mathcal{G}$ the observations are exchangeable, we can randomly permute $J$ times the class labels to compute the null distribution of $\mathcal{S}$, i.e., we consider the transformation $\mathbf{P}_j \mathbf{Y}$ \cite{commenges2003transformations}. We fix as first transformation $\mathbf{P}_1$ the identity one to get exact $\alpha$ control \cite{pesarin2010permutation,hemerik2018exact}, i.e., $\mathcal{S}_1$ is the observed statistical test. Let consider the statistical test $\mathcal{S}$ computed under transformation $j$ of the data as $\mathcal{S}_j$ with $j \in \{1, \dots, J\}$, the $p$-value is simply calculated as
\begin{equation}\label{pvalue}
    p = \dfrac{\sum_{j = 1}^{J} \mathbb{I}_{\mathcal{S}_j \ge \mathcal{S}_1}} {J}
\end{equation}
considering a right-tailed rejection region. If the $p$-value is less than the given significance level $\alpha$, we declare $\mathcal{F} \ne \mathcal{G}$. 

\section{Power analysis}\label{power}

Since PLS-based methods are not based on statistical distributions and are not likelihood estimation-based methods, strategies for power analysis should be based on numerical simulation. We propose here an approach to simulate data under the alternative hypothesis (Subsection \ref{model_h1}) and the full procedure to estimate power and sample size (Subsection \ref{sample}), considering the statistical tests presented in Section \ref{mcc}.

\subsection{Simulate data under the alternative hypothesis}\label{model_h1}

The PLS model of the pilot data is used to simulate new datasets with a given sample size $\tilde{N}$, which are in turn used to calculate the power of the statistical test $\mathcal{S}$ defined in Section \ref{mcc}. The effect size is assumed to be the same captured by the PLS model of the pilot data and is not modified during the simulation. Moreover, since PLS techniques exploit the correlation structure underlying $\mathbf{X}$, the covariance structure of the pilot data should be preserved when new data are simulated under the alternative hypothesis.

In the following, we define the proper simulation model.

\begin{definition}\label{sim}
    Considering the PLS model of Definition \ref{bifact} and the pilot data $\mathbf{X}$ and $\mathbf{Y}$, the matrix $\tilde{\mathbf{X}}$ of the simulated data under the alternative hypothesis $\mathcal{H}_1$ is defined as:
\begin{equation*}
    \tilde{\mathbf{X}} = \tilde{\mathbf{T}} \mathbf{P}^\top + \tilde{\mathbf{E}}
\end{equation*}
where $\tilde{\mathbf{T}}$ is the score matrix under $\mathcal{H}_1$, $\mathbf{P} = \mathbf{X}^\top \mathbf{T}(\mathbf{T}^\top \mathbf{T})^{-1}$ is the loading matrix calculated by the PLS model of $\mathbf{X}$ and $\mathbf{Y}$, and $\tilde{\mathbf{E}}$ is the simulated residual matrix.
\end{definition}

It is worth noting that the number of observations of the simulated data can differ from that of the pilot data. The following theorem defines the constraints that $\tilde{\mathbf{T}}$ and $\tilde{\mathbf{E}}$ must satisfy in order to preserve the covariance structure of the pilot data $\mathbf{X}$.

\begin{theorem}\label{thm}
Considering the model of Definition \ref{sim} and that of Definition \ref{bifact}, under the assumptions $\Vert \mathbf{T} \mathbf{P}^\top \Vert_{F} \gg \Vert \mathbf{E} \Vert_{F}$ and $\Vert \tilde{\mathbf{T}} \mathbf{P}^\top \Vert_{F} \gg \Vert \tilde{\mathbf{E}} \Vert_{F}$, if $\tilde{\mathbf{T}}^\top \tilde{\mathbf{T}} = \mathbf{T}^\top \mathbf{T}$ and $\tilde{\mathbf{T}}^\top \tilde{\mathbf{E}} = \mathbf{0}$ then $\mathbf{X}^\top \mathbf{X} \approx \tilde{\mathbf{X}}^\top \tilde{\mathbf{X}}$.
\end{theorem}
\begin{proof}
Considering the PLS model, since $\mathbf{P} \mathbf{T}^\top$ is orthogonal to $\hat{\mathbf{E}}$, the covariance matrix of $\mathbf{X}$ equals
\begin{equation*}
    \text{cov}(\mathbf{X}) \propto \mathbf{X}^\top \mathbf{X} = \mathbf{P} \mathbf{S} \mathbf{P}^\top + \mathbf{E}^\top \mathbf{E} 
\end{equation*}
where $\mathbf{S} = \mathbf{T}^\top \mathbf{T}$, that is approximately $\mathbf{P} \mathbf{S} \mathbf{P}^\top$ when
\begin{equation}\label{eq:cond}
    \Vert \mathbf{T} \mathbf{P}^\top \Vert_{F} \gg \Vert \mathbf{E} \Vert_{F}.
\end{equation}
If some data structures are still present in the residual matrix, condition \eqref{eq:cond} may not be satisfied. In this case, the residual matrix can be modeled by PCA, and the obtained scores and loadings can be included in the matrix factorization of $\mathbf{X}$ generated by PLS.

Analogously, considering Definition \ref{sim}, one has
\begin{equation*}
    \text{cov}(\tilde{\mathbf{X}}) \propto \tilde{\mathbf{X}}^\top \tilde{\mathbf{X}} = \mathbf{P} \tilde{\mathbf{T}}^\top \tilde{\mathbf{T}} \mathbf{P}^\top + \tilde{\mathbf{E}}^\top \tilde{\mathbf{E}} + \tilde{\mathbf{E}}^\top \tilde{\mathbf{T}} \mathbf{P}^\top + \mathbf{P} \tilde{\mathbf{T}}^\top \tilde{\mathbf{E}}. 
\end{equation*}
If $\tilde{\mathbf{T}}^\top \tilde{\mathbf{T}} = \mathbf{S}$ and $\tilde{\mathbf{T}}^\top \tilde{\mathbf{E}} = \mathbf{0}$, under the condition 
\begin{equation*}
    \Vert \tilde{\mathbf{T}} \mathbf{P}^\top \Vert_{F} \gg \Vert \tilde{\mathbf{E}} \Vert_{F}
\end{equation*}
the simulated data $\tilde{\mathbf{X}}$ and the pilot data $\mathbf{X}$ show the same covariance structure.
\end{proof}

To guarantee that $\tilde{\mathbf{T}}^\top \tilde{\mathbf{T}} = \mathbf{T}^\top \mathbf{T}$ and $\tilde{\mathbf{T}}^\top \tilde{\mathbf{E}} = \mathbf{0}$, the following procedure is here proposed to simulate a set of $\tilde{N}$ observations. After to have estimated the multivariate distribution of the PLS-scores in $\mathbf{T} \in \mathbb{R}^{N \times A}$ (including also the scores of the PCA model of the residuals if needed) by kernel density estimation-based approaches \cite{sheather2004density}, the elements of the matrix $\tilde{\tilde{\mathbf{T}}} \in \mathbb{R}^{\tilde{N} \times A}$ are calculated sampling that distribution. In general, the scores in $\tilde{\tilde{\mathbf{T}}}$ are not orthogonal, and its covariance structure is different from that of $\mathbf{T}$. As a consequence, a suitable orthogonalization procedure must be applied. Considering the SVD $\tilde{\tilde{\mathbf{T}}} (\mathbf{T}^\top \mathbf{T}) = \mathbf{U} \mathbf{D} \mathbf{V}^\top$, the matrix $\tilde{\mathbf{T}} \in \mathbb{R}^{\tilde{N} \times A}$ is calculated by
\begin{equation*}
\tilde{\mathbf{T}} = \mathbf{U} \mathbf{V}^\top (\mathbf{T}^\top \mathbf{T})^{1/2}.
\end{equation*}

Thus, the residual matrix $\tilde{\mathbf{E}}$ is calculated as follows. The rows of the residual matrix of the PLS model $\hat{\mathbf{E}}$ (or those of the residual matrix after PCA modeling of the PLS-residual matrix if needed) are sampled with replacement $\tilde{N}$ times to obtain the rows of the new matrix $\tilde{\tilde{\mathbf{E}}}$ that is made orthogonal to $\tilde{\mathbf{T}}$ by projection as
\begin{equation*}
\tilde{\mathbf{E}} = (\mathbf{I}_{n} - \tilde{\mathbf{T}}(\tilde{\mathbf{T}}^\top \tilde{\mathbf{T}})^{-1} \tilde{\mathbf{T}}^\top) \tilde{\tilde{\mathbf{E}}}.
\end{equation*}
Once $\Vert \tilde{\mathbf{T}} \mathbf{P}^\top \Vert_{F} \gg \Vert \tilde{\mathbf{E}} \Vert_{F}$ is numerically tested, it can be proved that the proposed procedure leads to scores and residuals that satisfy Theorem \ref{thm}.

The procedure here introduced is general and can be used both for regression and for classification problems.

The estimation of the dependent variables in $\Tilde{\mathbf{Y}}$ associated with the $\Tilde{N}$ observations and covariates $\Tilde{\mathbf{X}}$ depends in general on the support $\mathcal{Y}$. In the case of $\mathcal{Y} = \{1, \dots, G\}$ with $G>2$, the class of the new observations can be assessed on the basis of the distributions in the score space of the pilot data, partitioning that space by class. For instance, in the simple case of $\mathcal{Y} = \{1,2\}$, new observations for a given class are simulated by sampling the score distribution of the observations of that class for the pilot data. In the case of $\mathcal{Y} = \mathbb{R}$, the dependent variable may be estimated using the PLS model of the pilot data to predict the new simulated observations, adding an error term calculated sampling the distribution of the error term of the pilot data, but this case is out of the aim of the present study.


\subsection{Power and sample size calculation}\label{sample}

Given the procedure that allows the simulation of new data under the alternative hypothesis (Section \ref{model_h1}), the statistical test $\mathcal{S}$ introduced in Section \ref{mcc}, and assuming a significance level $\alpha$ and a sample size $\Tilde{N}$, the power is estimated applying the pseudocode defined in Algorithm \ref{alg:power1}.

\begin{algorithm}[H]
\caption{The pseudocode shows the procedure to estimate the power of a PLS model with $A$ score components considering a dataset with $\Tilde{N}$ observations given the significance level $\alpha$, the number of simulations $I$ used in MC simulation and the number of permutations $J$ used to estimate the $p$-value of the statistical test $\mathcal{S}$ defined in Section \ref{mcc}} \label{alg:power1}
\begin{algorithmic}[1]
\Require $\mathbf{X}$, $\mathbf{Y}$, $A$, $\Tilde{N}$, $\mathcal{S}$, $\alpha$, $I$, $J$ \Comment{$\mathbf{X}$, $\mathbf{Y}$ are the pilot data.}
\Ensure $\text{power}$
\State $\text{power} \leftarrow 0$
\For{$i$ in $1,\dots, I$}
\State $\tilde{\mathbf{X}}$, $\tilde{\mathbf{Y}}$ $\leftarrow$ simulate$(\mathbf{X}, \mathbf{Y}, \Tilde{N})$ \Comment{see Section \ref{model_h1} }
\State out $\leftarrow$ PLS$(\tilde{\mathbf{X}}, \tilde{\mathbf{Y}}, A)$ \Comment{compute the PLS model} \label{plsstep}
\State Compute $\mathcal{S}_1, \dots, \mathcal{S}_J$ \Comment{null distribution of $\mathcal{S}$ using results from step \ref{plsstep}}
\State $p = \frac{|\mathcal{S}_j \ge \mathcal{S}|} {J}$ \Comment{compute $p$-value} \label{line:pv}
\If{$Ap\le \alpha$}
\State $\text{power} \leftarrow \text{power} + 1/I$ \Comment{compute power}
\EndIf
\EndFor
\end{algorithmic}
\end{algorithm}

The procedure is general and can be applied both to classification and to regression problems once suitable statistical tests are introduced.

It is worth noting that considering a PLS model with more than $1$ score component, the permutation-based $p$-values (described in Equation \eqref{pvalue} and calculated in row \ref{line:pv} of Algorithm \ref{alg:power1}) must be corrected for multiplicity to control the family-wise error rate (FWER) \cite{goeman2014multiple}. Indeed, considering $A$ score components and the Bonferroni method, the adjusted p-value is $\tilde{p}_A = A p_A$ where $p_A$ is the $p$-value related to $\mathcal{H}_0: \mathcal{F}=\mathcal{G}$ when $A$ score components are considered.

Given the procedure for power calculation defined in Algorithm \ref{alg:power1}, the sample size estimation can be performed following the procedure described in the pseudocode of Algorithm \ref{alg:power2}.

\begin{algorithm}[H]
\caption{The pseudocode shows the procedure to estimate the optimal sample size $\hat{N}$ for a PLS model with $A$ score components given the significance level $\alpha$ and power level $1-\beta$; $I$ is the number of simulations used in power calculation and $J$ the number of permutations used to estimate the $p$-value of the statistical test $\mathcal{S}$ defined in Section \ref{mcc}. The algorithm takes as initial candidate value $N_{\min}$.}\label{alg:power2}
\begin{algorithmic}[1]
\Require $\mathbf{X}$, $\mathbf{Y}$, $A$, $N_{\min}$, $\mathcal{S}$, $\alpha$, $\beta$, $I$, $J$ \Comment{$\mathbf{X}$, $\mathbf{Y}$ are the pilot data.}
\Ensure $\hat{N}$
\State $n\leftarrow N_{\min}$
\State$\text{power}(n)\leftarrow \text{calculate power} (\mathbf{X},\mathbf{Y},$A$,n,\mathcal{S},\alpha,I,J)$ \Comment{use Algorithm \ref{alg:power1}}
\While{$\text{power}(n) \ge 1-\beta$}
\State $n \leftarrow n+1$
\State$\text{power}(n)\leftarrow \text{calculate power} (\mathbf{X},\mathbf{Y},$A$,n,\mathcal{S},\alpha,I,J)$ \Comment{use Algorithm \ref{alg:power1}}
\EndWhile
\State $\hat{N} \leftarrow n$
\end{algorithmic}
\end{algorithm}


\section{Applications to data}\label{applications}
Simulated data and a real dataset are here investigated. Calculations were performed using a scientific computing cluster with a processor having 20 CPU and 200 GB of RAM and the R package called \texttt{powerPLS} available on GitHub (\url{https://github.com/angeella/powerPLS}).

\subsection{Simulated data}\label{simulated data}

The main advantage of considering simulated data is that their structure is a-priori known. This study considers a 2-class problem with a pilot dataset composed of $N_g = 5 \,\, \forall g \in \{1,2\}$ observations per class and $P = 30$ covariates. The dataset has been built imposing the following data structure: $5$ covariates closely related to the class membership and $25$ noisy covariates. Specifically, the matrix $\mathbf{X}\in \mathbb{R}^{10 \times 30}$ of the pilot data has been simulated as the block matrix
\begin{equation}\label{pilot}
    \mathbf{X} = [\mathbf{T}_{\text{pilot}} \mathbf{P}_{\text{pilot}}^\top |\mathbf{X}_{\bar{R}}]
\end{equation}
where $[\mathbf{T}_{\text{pilot}} \mathbf{P}_{\text{pilot}}^\top]\in \mathbb{R}^{10 \times 5}$ is associated to the class and $\mathbf{X}_{\bar{R}}\in \mathbb{R}^{10 \times 25}$ contains random noise. The matrix $\mathbf{T}_{\text{pilot}} \in \mathbb{R}^{10 \times A_{\text{pilot}}}$ is the PCA-score matrix of $C = [C_1^\top | C_2^\top]$ where $C_1 \sim \mathcal{MN}(0,\mathbf{I}_{A_{\text{pilot}}})$ and $C_2 \sim \mathcal{MN}(\mu,\mathbf{I}_{A_{\text{pilot}}})$ with $\mu \in \{2,5\}$, $\mathbf{P}_{\text{pilot}} \in \mathbb{R}^{5 \times A_{\text{pilot}}}$ is the PCA-loading matrix of a $(A_{\text{pilot}} \times 5)$ matrix sampled from a $\mathcal{U}(0,1)$, and $\mathbf{X}_{\bar{R}}$ is sampled from a $\mathcal{MN}(\mathbf{0},\mathbf{I}_{25})$. The parameter $\mu$, which defines the distance between the centers of the distributions of the two classes, is used to set the effect. Indeed, large values of $\mu$ can be interpreted as large effects. The dimension $A_{\text{pilot}}$ has been set to $2$ (the results for $A_{\text{pilot}}=3$ are reported in Appendix \ref{appendix:AppendixC}).  

Figure \ref{fig:mcc} shows the mean of the estimated power across $30$ simulations of the pilot data following the procedure defined in Equation \eqref{pilot}. For each of the $30$ simulations, the power has been estimated using the statistical tests specified in Section \ref{mcc}, $J = 200$ random permutations for the $p$-value estimation, and $I = 100$ simulations for the data under the alternative hypothesis. The number of observations per class was $N_1 = N_2 \in \{5, 10, 15, 20, 25, 30\}$, and the number of score components $A \in \{1,2,3,4\}$. The upper panel of Figure \ref{fig:mcc} refers to the results using MCC as statistical test, the center panel to the results obtained considering the squared Pearson correlation coefficient $R^2$ between the dependent variable mapped into the Euclidean space and the estimated one, and the bottom panel to the results when the statistical test based on $\mathbf{T}_p$ is used in the power analysis estimation process.

\begin{figure}
    \centering
    \includegraphics[width = 1\textwidth]{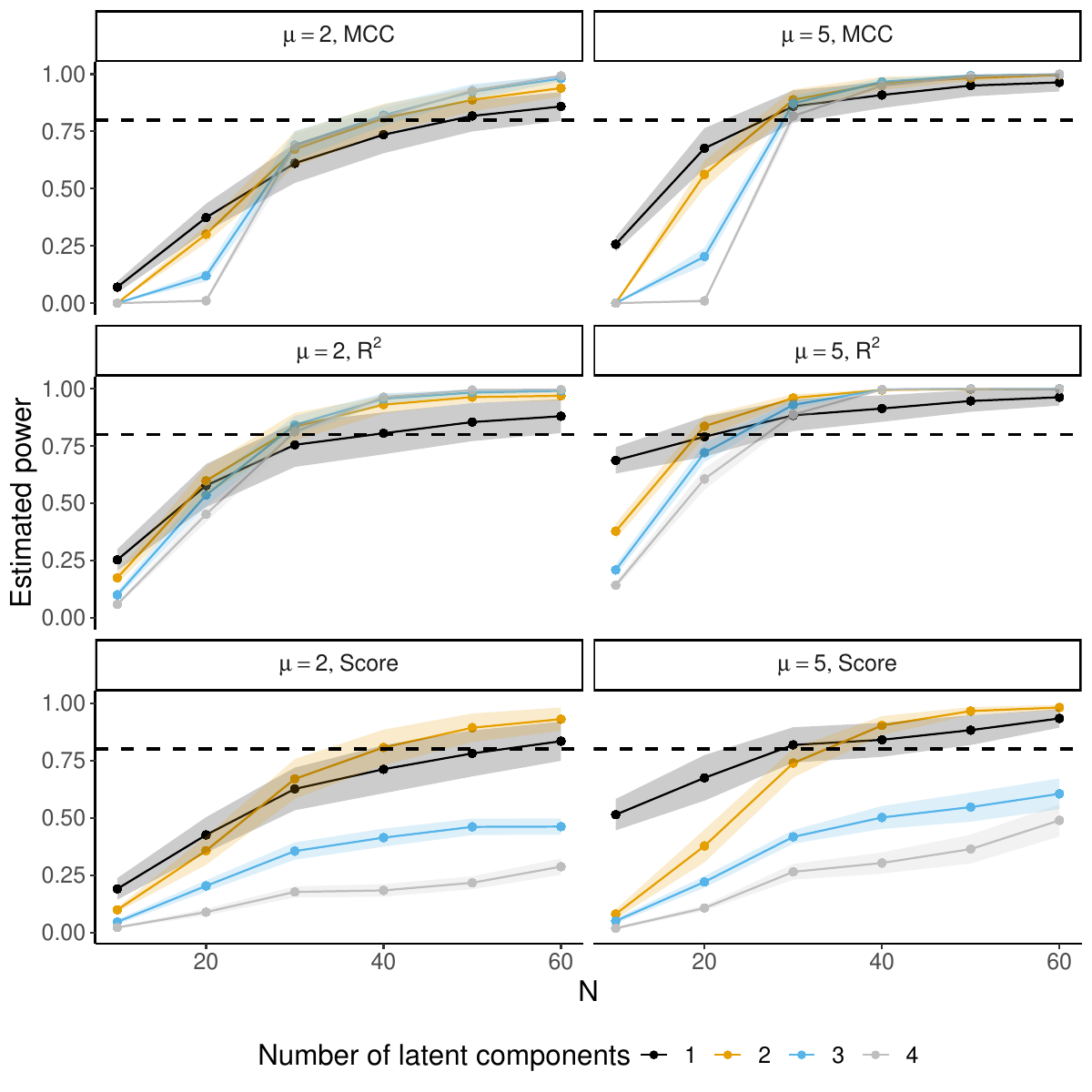}
    \caption{Estimated power across different sample size $N \in \{10,20,30,40,50,60\}$ using the statistical tests $\mathcal{S}$ introduced in Section \ref{mcc}. The pilot data have been simulated with $A_{\text{pilot}}=2$, whereas the power has been estimated considering a PLSc model with $A \in \{1,2,3,4\}$ score components. A different color has been used to represent the power curves for different $A$. The colored shadow areas represent the corresponding confidence intervals at level $0.95$. $100$ MC simulations and $200$ permutations have been considered for each analysis.}
    \label{fig:mcc}
\end{figure}

Figure \ref{fig:mcc} shows that power increases with increasing sample size, as expected. Moreover, in the case of small sample sizes, the MCC-based statistical test has lower power than the other two tests analyzed. For example, fixing $N = 10$ and considering the largest effect size (i.e., $\mu = 5$) and one score component, the mean across $30$ simulations of the estimated power equals $0.257$ if MCC is considered, whereas it equals $0.514$, $0.687$ if the score-based and the $R^2$ are employed, respectively. This result confirms the intuition of Rosenblatt and coauthors \cite{rosenblatt2021better}, where the discretization nature of the MCC-based test leads to a loss of power compared to test statistics that exploit the continuous aspect of the classification model under analysis. It is well known that using permutation theory to compute the null distribution of a discrete statistical test is conservative \cite{hemerik2018exact,rosenblatt2021better}. However, the conservativness generally decreases with sample sizes \cite{kim2016classification} as we can note in Figure \ref{fig:mcc}.

The behavior of the power curve is not only dependent on the statistical test but also heavily dependent on the number of score components of the PLSc model. For MCC and $R^2$, the curves are similar in the region of large sample size, independently of the number of score components, but show differences in the area of small sample size. Indeed, for small sample sizes, PLSc may overfit the data modeling noise with the effect of decreasing power. This effect decreases with the increase in the sample size. For the score-based statistic, the curves are different, and the behavior is more complex.
The same trends are observed in Figure \ref{fig:mcc_A3} of Appendix \ref{appendix:AppendixC} obtained considering pilot data with $A_{\text{pilot}}=3$.

Independently of the statistical test, sample size estimation resulted to be heavily dependent on the number of score components, which must be accurately determined if one wants to obtain a reliable sample size estimation.

\subsection{Real data}\label{real data}

The dataset has been extracted from the data published in \cite{locci20191} where $59$ post-morten aqueous humour samples were collected from closed and opened sheep eyes. Each sample was analyzed by 1H NMR spectroscopy, obtaining the quantification of $43$ metabolites. As a result, a pilot dataset composed of $59$ observations ($29$ from opened eyes and $30$ from closed ones) and $43$ covariates was obtained. More details about sample collection, experimental procedure, and data preprocessing can be found in \cite{locci20191}. Data were autoscaled prior to performing data analysis. In power analysis, the residuals of the PLSc model were submitted to PCA to obtain a score structure able to explain at least the $80\%$ of the total variance of the pilot data.

The two groups of samples, corresponding to opened and closed eyes, were investigated by PLSc. Considering $A \in \{1,2,3,4\}$ score components, the models with the test statistics reported in Table \ref{table1} were obtained. The adjusted p-values were less than $0.05$ for all the statistics. As expected, all the statistics increased with the increase in the number of components.

The power was greater than 0.90 for each model, independently of the statistical test used. During the data simulation, the level of similarity between the pilot data and the data simulated under the alternative hypothesis was assessed by computing two measures of association between matrices, i.e., the RV coefficient \cite{escoufier1973traitement} and the Procrustes one \cite{gower1971statistical}. Both indices take values in $[0,1]$, where $0$ stands for the absence of association (i.e., orthogonal information) while $1$ equals complete similarity between the two data matrices. Considering the whole set of simulated data, the RV index was at least equal to $0.923$, while the Procrustes one showed a minimum equal to $0.919$. These results proved that the covariance structure was preserved during power estimation.

Calculating the power curves for ${6, 12, 20, 30, 42}$ observations per class, considering $J = 200$ random permutations for the p-value estimation and $I = 100$ simulations for the data under the alternative hypothesis, we found that the power curves obtained considering $R^2$ showed greater power than those of the other two statistics. Specifically, considering, for instance, a PLSc model with two score components, a sample size of approximately 24 observations per class is required for MCC and score-based statistics to have a power of $0.80$, whereas a sample size of approximately 16 observations per class is requested for $R^2$.

\begin{table}

\caption{Real dataset: test statistics (as defined in Section \ref{mcc}) for the pilot data considering $A \in \{1,2,3,4\}$ number of score components.}\label{table1}
\centering
\begin{tabular}[t]{rrrr}
\toprule
$A$ & \textbf{MCC} & \textbf{Score} & $R^2$\\
\midrule
1 & 0.83 & 10.4 & 0.66\\
2 & 0.87 & 12.7 & 0.74\\
3 & 0.93 & 14.8 & 0.79\\
4 & 0.97 & 15.7 & 0.81\\
\bottomrule
\end{tabular}
\end{table}

\section{Concluding remarks}\label{discussion}
We have introduced an innovative procedure for conducting power analysis within the context of PLS-based methods.

The proposed approach leverages the score structure identified in the pilot data when simulating data under the alternative hypothesis to estimate power across varying sample sizes. It considers explicitly the data decomposition discovered by PLS and can be applied in principle both to regression and to classification problems. Following the strategy introduced in Section \ref{model_h1}, the correlation structure of the pilot data has been preserved during data simulation, as proved by investigating the real data set in Section \ref{real data}.

For the sake of simplicity, 2-class classification problems were investigated, testing the null hypothesis of no differences between classes. Specifically, we have introduced three permutation-based statistical tests to analyze the covariate distribution between the two classes. The approach overcomes cross-validation issues when analyzing pilot data with a small sample size (i.e., when less than 10-15 observations per class are available), even if for larger pilot data, cross-validation may be used to estimate the test statistics to use in statistical tests.

To evaluate the effectiveness of our proposed power analysis approach, we conducted simulations across various scenarios and analyzed a real data set.

In all cases, the power curve increased with the increase in the sample size, as expected. Interestingly, the power curve seems to depend heavily on the number of score components used in PLS. Consequently, estimating the correct number of scores to use in PLS modeling is fundamental to obtaining reliable power and sample size estimation. In principle, if statistical tests based on the estimated matrix of the regression coefficients are used, i.e., MCC and $R^2$, the greater the number of score components is, the greater the differences detected between classes are, i.e., MCC and $R^2$ increase, but then power may not increase. Indeed, over-fitting may be present when an excessive number of score components is used, increasing the p-values estimated under the null hypothesis and decreasing the power. Moreover, the correction for FWER may limit the effect of increasing the number of scores, reducing the significance level of the test. From the simulated data and the real dataset, the score-based statistical test seems to show a slightly different behavior with respect to MCC and $R^2$, probably because the differences between the centers of the classes and the dispersion around these centers in the predictive score space are differently weighted in the test statistic. Again, $R^2$ seems to be a better candidate than MCC since it increases the power of the test, at least in the case of small pilot data sets. However, $R^2$ may be misleading in the case of classification problems because small and large residuals in the calculation of the dependent variable may be associated with the same class, making $R^2$ an unreliable parameter to measure the goodness in classification, i.e., small $R^2$ may be associated to large MCC. Some type of regularization could be necessary to adapt $R^2$ to classification. Moreover, both MCC and $R^2$ can be used to study more general multi-class problems. Still, new score-based statistical tests must be introduced for a general $G$-class problem since $G-1$ predictive scores are calculated.

The present study must be considered a preliminary study since it does not address all the issues of power analysis, even if it draws a methodology towards a comprehensive approach. Moreover, it is worth noting that the same lines of thought presented for the 2-class classification problem can be adapted to deal with more complex problems.

The main limit of the study is that the effect size has not been considered as a parameter to be investigated in power analysis. Indeed, since the new data were simulated to preserve the correlation structure of the pilot data, the effect size was maintained unchanged during the power calculation. It is not trivial how to define and measure the effect size in PLS-based methods; a dedicated study will deal with this topic. However, the data decomposition in predictive and non-predictive parts generated by PLS also paves the way for the possibility of defining and modifying the effect size for a more general power analysis. Indeed, a natural approach may be changing the predictive score structure to increase or decrease the effect size, leaving the non-predictive part unchanged, but this will be discussed in a further study.

Another limitation is that the power in the estimation of the number of PLS-score components and that of the relevant features discovered by PLS were not considered here. Even if these two points will be discussed in the future, we want to disclose that the methodology proposed here can also be adapted to address these issues.


\section*{Acknowledgments}

Angela Andreella gratefully acknowledges funding from the grant BIRD2020/SCAR ASSEGNIBIRD2020\_01 of the University of Padova, Italy, and PON 2014-2020/DM 1062 of the Ca’ Foscari University of Venice, Italy. Some of the computational analyses done in this manuscript were carried out using the Ca' Foscari University of Venice multiprocessor cluster ``SCSCF: Ca' Foscari Scientific Computing System'', \url{https://www.unive.it/pag/30351/}.

\section*{Authors Contributions}
\textbf{AA}: Conceptualization, methodology, software, formal analysis, investigation, writing - original draft; \textbf{LF}: conceptualization, methodology, supervision, writing - review \& editing; \textbf{BS}:  conceptualization, methodology, supervision, writing - review \& editing, \textbf{MS}: Conceptualization, methodology, software, formal analysis, investigation, writing - original draft.

\section*{Declaration of Competing Interest}
The authors declare no competing interests.

\newpage

\appendix

\section{PLS algorithm}\label{appendix:AppendixA}
Several algorithms have been proposed for PLS regression. Here, we report the so-called ``eigenvalue'' PLS2 algorithm fixing the number of latent components equals $A$.
\begin{algorithm}[H]
\caption{``eigenvalue'' PLS2 algorithm}
\begin{algorithmic}[1]
\Require $\mathbf{X} \in \mathbb{R}^{N \times P}$; $\mathbf{Y} \in \mathbb{R}^{N \times K}$; $A$
\State $\hat{\mathbf{E}}_{0} = \mathbf{X}$
\State $\hat{\mathbf{F}}_{0} = \mathbf{Y}$
\For{$a$ in $1,\dots, A$}
\State $\hat{\mathbf{E}}_{a-1}^\top \hat{\mathbf{F}}_{a-1} \hat{\mathbf{F}}_{a-1}^\top \hat{\mathbf{E}}_{a-1} \mathbf{w}_a = \lambda_a \mathbf{w}_a$ \label{lst:pls2} \Comment{Estimate $\mathbf{w}_a$}
\State $\mathbf{t}_a = \hat{\mathbf{E}}_{a-1} \mathbf{w}_a$
\State $\mathbf{Q}_{\mathbf{t}_a} = \mathbf{I}_{N} - \mathbf{t}_a (\mathbf{t}_a^\top \mathbf{t}_a)^{-1} \mathbf{t}_a^\top$ 
\State $\hat{\mathbf{E}}_a = \mathbf{Q}_{\mathbf{t}_a} \hat{\mathbf{E}}_{a-1}$  \Comment{$\mathbf{X}$-deflation step }
\State $\hat{\mathbf{F}}_a = \mathbf{Q}_{\mathbf{t}_a} \hat{\mathbf{F}}_{a-1}$ \Comment{$\mathbf{Y}$-deflation step }
\EndFor
\end{algorithmic}\label{algo:pls2}
\end{algorithm}
The matrices $\hat{\mathbf{E}}_a$ and $\hat{\mathbf{F}}_a$ are called the residual matrix of the $\mathbf{X}$- and $\mathbf{Y}$-block, respectively, the vectors $\mathbf{w}_a$ and $\mathbf{t}_a$ are called weight vector and score vector, respectively, and $\mathbf{Q}_{\mathbf{t}_a}$ is an orthogonal projection matrix that projects a given vector into the space orthogonal to $\mathbf{t}_a$. We denote $\hat{\mathbf{E}} \coloneqq \hat{\mathbf{E}}_A$ and $\hat{\mathbf{F}} \coloneqq \hat{\mathbf{F}}_A$.
When the calculation of the weight vector in step \ref{lst:pls2} of Algorithm \ref{algo:pls2} is replaced by a given vector $\mathbf{w}_a$ defined as input of the algorithm, the algorithm becomes the ``Iterative Deflation Algorithm'' (IDA), which is a general algorithm able to solve the least squares problem
\begin{equation*}
        \mathbf{\hat{B}}_{\text{LS}}=\arg\min_{\mathbf{B}}|| \mathbf{Y} - \mathbf{X} \mathbf{B}||_{F}^2 
\end{equation*}
for a non-trivial choice of the vectors $\mathbf{w}_a$ \cite{stocchero2019iterative}. The main properties of the ``eigenvalue'' PLS2 algorithm have been extensively discussed in the past literature. Readers can refer to \cite{Hoskuldsson1988} and to \cite{stocchero2019iterative}.

\section{Post-transformation of PLS2}\label{appendix:AppendixB}
Post-transformation of PLS2 has been introduced in \cite{stocchero2016post} to separate the structured data variation discovered by PLS2 into the predictive and non-predictive parts. From a geometrical point of view, post-transformation linearly transforms the score space of the PLS2 model spanned by $\mathbf{T}$ to obtain two new sets of scores: the predictive scores $\mathbf{T}_{P}$ able to explain the dependent matrix $\mathbf{Y}$, and the non-predictive scores $\mathbf{T}_{o}$ that are orthogonal to $\mathbf{Y}$, i.e.
\begin{equation*}
    [\mathbf{T}_{P} \mathbf{T}_{o}] = \mathbf{T} \Tilde{\mathbf{G}}
\end{equation*}
where the score matrices $\mathbf{T}_{P} = [\mathbf{t}_{P}]$ and $\mathbf{T}_{o} = [\mathbf{t}_{o}]$ have been introduced, and the matrix $\Tilde{\mathbf{G}}$ is a suitable non-singular matrix.

Post-transformation is performed using the columns of $\mathbf{W}\mathbf{G}$, where $\mathbf{G}$ is a suitable orthogonal matrix, as weight vectors within the IDA instead of the weight vectors $\mathbf{w}_a$ calculated by PLS2. The algorithm to calculate the matrix $\mathbf{G}$ required to post-transform the PLS2 model is reported in Algorithm \ref{algo:pt}.

\begin{algorithm}[H]
\caption{Algorithm to calculate the matrix $G$}\label{algo:pt}
\begin{algorithmic}[1]
\Require $\mathbf{X} \in \mathbb{R}^{N \times P}$; $\mathbf{Y} \in \mathbb{R}^{N \times K}$; $\mathbf{W} \in \mathbb{R}^{P \times A}$
\Ensure $\mathbf{G} \in \mathbb{R}^{A \times A}$
\State $\mathbf{Y}^\top \mathbf{X} \mathbf{W} = \mathbf{U} \mathbf{S} \mathbf{V}^\top$ \Comment{Singular Value Decomposition}
\State $(\mathbf{I}_{A} - \mathbf{V} \mathbf{V}^\top) \mathbf{g}_{o_i}=\lambda_{o_i} \mathbf{g}_{o_i}$ \Comment{$M$ positive eigenvalues}
\State $\mathbf{G}_{o}=[\mathbf{g}_{o_1} \cdots \mathbf{g}_{o_M}]$
\State $(\mathbf{I}_{A} - \mathbf{G}_{o} \mathbf{G}_{o}^\top) \mathbf{g}_{P_i}=\lambda_{P_i} \mathbf{g}_{P_i}$ \Comment{${A-M}$ positive eigenvalues}
\State $\mathbf{G}=[\mathbf{G}_{o} \mathbf{g}_{P_1} \cdots \mathbf{g}_{P_{A-M}}]$
\end{algorithmic}
\end{algorithm}

Alternative algorithms for calculating $\mathbf{G}$ have been proposed \cite{stocchero2016post}. Moreover, since all the models based on IDA can be post-transformed \cite{stocchero2019iterative}, procedures of post-transformation can be, in principle, developed for most of the PLS methods because most of the PLS-based techniques are based on the IDA.

\section{Simulated data}\label{appendix:AppendixC}

Figure \ref{fig:mcc_A3} has been obtained considering the same simulation process presented in Section \ref{simulated data}, but setting the dimension $A_{\text{pilot}}=3$.

\begin{figure}
    \centering
    \includegraphics[width = 1\textwidth]{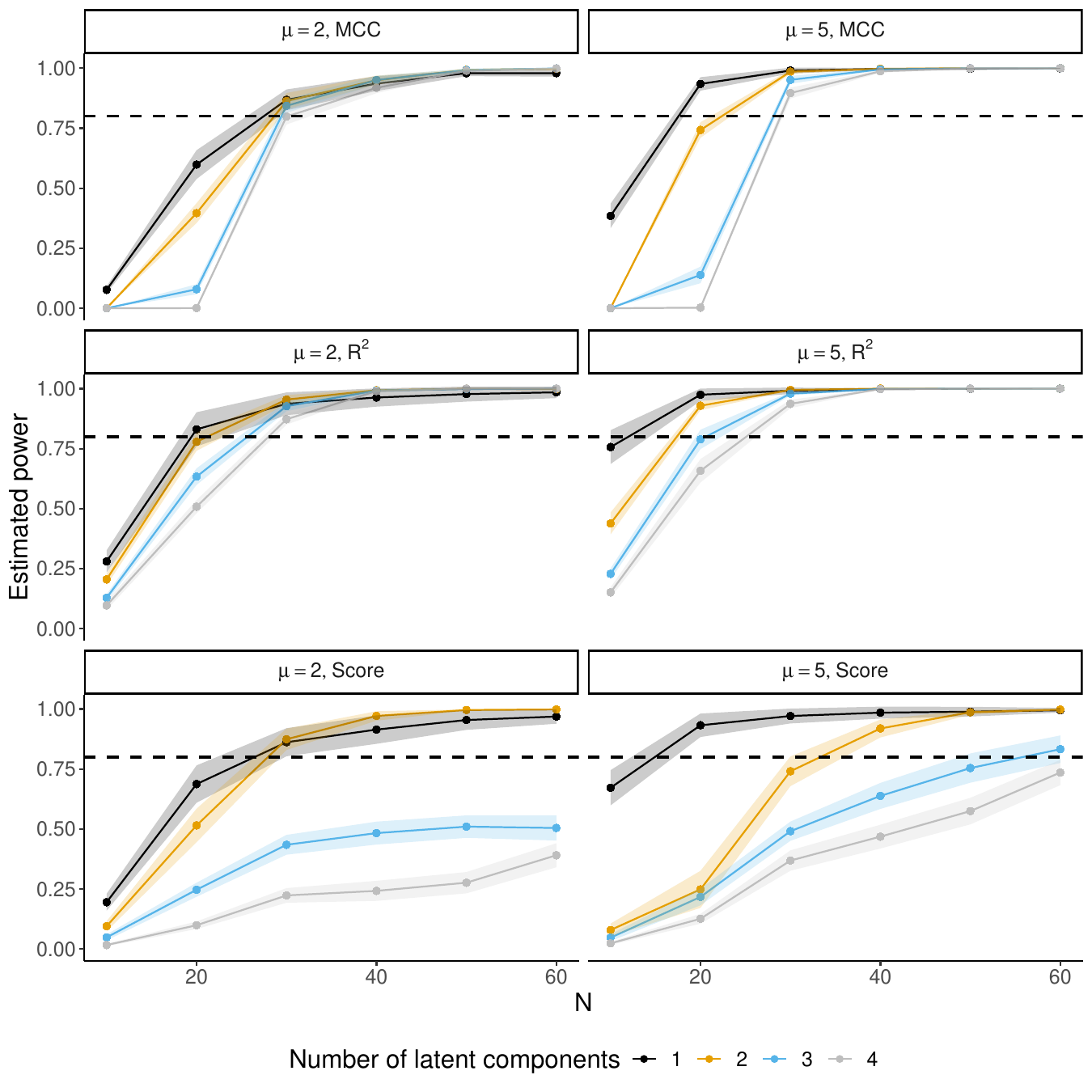}
    \caption{Estimated power across different sample size $N \in \{10,20,30,40,50,60\}$ using the statistical tests $\mathcal{S}$ introduced in Section \ref{mcc}. The pilot data have been simulated with $A_{\text{pilot}}=3$, whereas the power has been estimated considering a PLSc model with $A \in \{1,2,3,4\}$ score components. A different color has been used to represent the power curves for different $A$. The colored shadow areas represent the corresponding confidence intervals at level $0.95$. $100$ MC simulations and $200$ permutations have been considered for each analysis.}
    \label{fig:mcc_A3}
\end{figure}

\clearpage 

\bibliographystyle{apalike}
\bibliography{biblio}

\end{document}